\documentclass[11pt, reqno]{article}
\usepackage{eurosym}
\usepackage{amssymb}
\usepackage{amsmath}
\usepackage{amsfonts}
\usepackage[numbers,compress]{natbib}
\usepackage{graphicx}
\usepackage{color, colortbl}
\usepackage[table]{xcolor}
\usepackage{dblfloatfix}
\usepackage[flushleft]{threeparttable}
\usepackage{xcolor}
\usepackage{hyperref}

\newtheorem{corollary}{Corollary}
\newtheorem{assumption}{Assumption}
\newtheorem{theorem}{Theorem}
\newtheorem{remark}{Remark}
\begin{document}

\title{Optimal growth under irreversible pollution: The Forster-Ramsey model\thanks{\emph{In memory of Thierry Br\'echet}. This paper was started when Hritonenko and Yatsenko were visiting the Centre for Unframed Thinking (CUT), Rennes School of Business, in Summer 2022, and completed during the visit of Sahlin in Marseille in Spring 2026. We would like to thank Carmen Camacho for careful reading and useful suggestions. Boucekkine acknowledges funding from the French
government under the “France 2030” investment plan managed by the French National Research Agency (reference: ANR-17-
EURE-0020), and from the CNRS International Research Network grant $E3E$, 2025-2029. The usual disclaimer applies.} }
\author{
\parbox{35mm} {Raouf Boucekkine\footnote{Corresponding author. Aix Marseille Univ, CNRS, AMSE, Marseille, France; Senior chair, Institut Universitaire de France. E-mail:  raouf.boucekkine@univ-amu.fr
 }} \and
\parbox{35mm} {Natalia Hritonenko\footnote{Department of mathematics, Prairie View A\&M University, USA. E-mail: nahritonenko@outlook.com}} \and
\parbox{30mm} {Yuri Yatsenko\footnote{Department of mathematics, Houston Baptist University, USA. E-mail: yyatsenko@outlook.com}} \and
\parbox{30mm} {Ullrika Sahlin\footnote{Department of earth and environmental science, Lund University, Sweden. E-mail: ullrika.sahlin@mgeo.lu.se}}}
\date{\today }
\maketitle

\begin{abstract}
We study optimal growth in a Ramsey economy with irreversible pollution, where the natural absorption capacity is non‑monotonic and non-concave \emph{\`a la} Forster. 
In the absence of pollution control, the capital dynamics decouple from pollution, and we characterize a simple threshold on the maximal absorption capacity above which the optimal path is sustainable and below which an optimal regime of irreversible pollution  sets in. We then introduce an optimal abatement pollution control to assess whether irreversible pollution ceases to be optimal.
We derive a four‑dimensional dynamical system for capital, consumption, pollution and the (shadow) cost of pollution, coupled with an algebraic rule for optimal abatement. First, we show that while optimal abatement does shift the irreversibility threshold, it may not rule out optimal irrversible pollution paths. Second, and more importantly, we show that embedding the Forster's pollution dynamics into a Ramsey growth model generates a novel and robust hybrid ecological-economic irreversibility condition that is invisible in pollution-only frameworks. Coupling ecosystem curvature with the social discount rate, it is structurally robust and preeminent over all feasibility conditions amenable to policy intervention, and defines the boundary between a world where sustainability is achievable and one where it is not. 
\end{abstract}

\smallskip


\newpage

\section{Introduction}
As theorized by Forster (1975), and
developed later by Tahvonen and Withagen (1996), pollution irreversibility is reached
when the decay rate of pollution (or Nature's capacity to absorb pollution)
declines sharply above a certain level of pollution. Above this threshold, bifurcation
phenomena occur, formally generated by the non-concave nature of the embedded ecological
models. A leading example of such a phenomenon is the so-called shallow lake problem (see e.g,  Maler, 2000): 
small variations in phosphorus loads in such lake ecosystems may lead to the emergence of tipping points, eventually causing
significant losses in ecosystem services.

In the recent years, numerous cases of irreversibility have been reported in general science reviews. These new accounts for irreversible non-linear dynamics  are quite diverse, and concern either global or local pollution cases. On the global scale, Ke et al (2024) documented two years ago a lower capacity of land and oceans to serve as sinks for carbon emissions, which can be explained by increased emission from forest fires, drought and extreme heat, factors that are caused by GHG. The 2026 Summer season will certainly not alleviate these unfavorable circumstances. Chemical pollution is another documented case where pollution can be irreversible. For instance, the ongoing accumulation of plastic in the environment shows that the rate at which plastic pollution enters an area exceeds the rate of natural removal processes or cleanup actions (MacLeod et al 2021), pretty much in line with the early conceptualization of irreversible pollution in the 70s. The efficiency of decay of plastic pollution depends on biological degradation, which in turn can be negatively affected by the nano particles because of weathering. On the top of that, mixtures of chemicals with low biodegradability and long range transportation could generate a pollution stock of chemical cocktails (Persson et al 2013).

Of course, a great deal of the literature concerned with irreversible pollution is devoted to the following question: Are irreversible pollution avoidable? If so, how? In their pioneering work, Tahvonen and Withagen (1996)
demonstrated that avoiding the threat of irreversible pollution thresholds is not
granted, even in the frame of a benevolent regulator. Extending the latter to a differential game set-up, Boucekkine et al. (2023) proved more recently that even full cooperation among players cannot
always prevent this unpleasant outcome. Furthermore adding uncertainty through possible catastrophic stochastic shocks to  irreversible pollution regimes, Boucekkine et al. (2025) established that as the actual pollution level becomes closer to the irreversibility thresholds, decision makers consider that irreversible pollution is unavoidable and increase production and thus, pollution accordingly.\footnote{The concept of avoidability of irreversibile pollution regimes has been initially suggested and discussed by Clarke and Reed (1994).}

In light of these papers (see also Leandri and Tidball, 2019, for a more recent simulation-based analysis), the need for a much tighter control of Nature's absoption capacity has become urgent. This idea is pushed quite far by El Ouardighi et al (2014, 2020) who build up models allowing firms to control the decay rate of pollution. However the way the latter is endogenized in their models eventually
nullifies the non-concavity inherent in the original Forster problem, leading them to deal with
softer irreversibility constraints. It is worth pointing out at this stage that all the research works cited so far are not growth models in the sense that they only include a single state variable, the pollution stock (or environmental quality). \footnote{This remark also applies to Toman and Withagen (2000) and to Brock and Starett (2003).} Introducing growth means typically including another state variable, the capital stock, which accumulates along the growth process. However, it has been quite neatly shown by Jones and Manuelli (2001) that the pollution paths are conditioned by the level of development of countries, to consider the cross-country level for example. The main contribution of our paper is to explore the optimality of irreversible pollution paths in a Ramsey growth model with a central planner choosing both the pollution and development paths through the optimal investment and pollution abatement controls subject to the non-concave strong irreversibility constraint first posed by Forster (1975).

Indeed, the overwhelming majority of growth models introducing pollution only consider reversible pollution through a constant rate of decline of the pollution stock, featuring a constant Nature's capacity to absorb pollution. A noticeable exception is the major work by Prieur (2009) who found, in an overlapping-generations model \emph{\`a la} Jones and Pecchenino (1994),  that even though firms can optimally decide about their level of pollution abatement, this might not be enough to escape from the irreversible pollution regime in the long term. Concretely, he found that allowing firms to choose optimally their
maintenance/pollution abatement efforts will not always prevent the economy to get
into the irreversible regime under the Forster hard irreversibility
constraint.\footnote{%
Bonneuil and Boucekkine (2016) is a singular contribution to the area. It
uses viability theory to depict \emph{viable} policies to escape from the
irreversible regime in a standard pollution model: it turns out that it's
not always possible.} This said, there exist a quite rich literature studying the relationship between growth and pollution and the inherent tradeoffs in optimal or general equilibrium growth models with reversible pollution starting with textbooks expositions (see the earlier analysis in Barro and Sala-i-Martin, 2003, chapter 2) and the seminal work of Stokey (1998) addressing these issues in both neoclassical (exogenous) growth settings and with endogenous growth. Other remarkable works in the same vein include, among others, Bovenberg and Smulders (1996), Byrne (1997), Economides and Philippopoulos (2008), and more recently, Amigues and Durmaz (2019), the latter using a two-sector model of growth.

To the best of our knowledge, there is no published work combining optimal growth \emph{\`a la} Ramsey with irreversible pollution \emph{\`a la} Forster. One of the possible reasons this task has not been undertaken is its technical complexity: two state variables (stocks of capital and pollution) and non-concavity yielding multiple solutions. As outlined above, earlier literature combining optimal growth theory with ecological dynamics (Tahvonen and Withagen, 1996; Toman and Withagen, 2000; Jones and Manuelli, 2001, or Brock and Starett, 2003) did stress that the shape of the natural pollution absorption function is critical: if absorption capacity is of the Forster type, multiple steady states may arise, and pollution can escape to an irreversible regime. However, most of these works employ highly stylized, low‑dimensional models that abstract from the full Ramsey macroeconomic structure.

In this paper we embed irreversible pollution into a standard Ramsey growth model with Cobb‑Douglas production, log consumption utility and disutility of pollution.
We first analyze the benchmark case without any abatement.
There, the famous dichotomy of the Ramsey model applies: capital and consumption evolve independently of pollution, so the problem reduces to the classic Ramsey dynamics plus a one‑way coupled pollution equation.
We show that, depending on a simple threshold involving the maximal absorption capacity, either pollution converges to a finite steady state or it grows without bound, and when it grows, this explosive path remains optimal in the sense of the optimal control problem. This is due to the fact that within the Ramsey model set-up, if the pollution stocks diverges, it grows at a linear rate, therefore the objective function of the central planner remains bounded as long as it embeds exponential time discounting. The next natural question is whether introducing a pollution control instrument can eliminate the optimal irreversible regime.
We show that the corresponding Hamiltonian first‑order conditions yield  a four‑dimensional system for capital, consumption, pollution, and the shadow cost of pollution complemented by an algebraic equation for the optimal abatement rule.
We show that optimal abatement does \emph{not} eliminate the possibility of irreversible pollution in the full analytical case of linear abatement technology. Instead, it changes the nature of the irreversibility criterion. We show that environmental policy can expand the parameter space for sustainability, but it cannot overcome fundamental geophysical limits when the natural decay process itself is insufficiently responsive. 

Indeed, a key original and robust theoretical result of our work is that a sustainable steady state exists only if the natural absorption non-concave \(f(P)\) satisfies
\[
f'(P) + \rho>0
\]
for some positive pollution level \(P\), where \(\rho>0\) is the rate of time preference. Importantly enough, we also show that this ecological-economic condition is robust to nonlinear specifications of the abatement technology. 
Much of the existing literature on environmental policy in dynamic settings focuses on the design of instruments, such like taxes, abatement subsidies, emission caps, that can steer the economy toward a desirable steady state. Implicitly, this literature treats sustainability as an economic problem: given the right incentives, the economy can be made to converge. Our analysis qualifies this optimism in an important way. We show that the conditions for sustainability are not all of the same kind. Some are standard economic feasibility conditions, which are indeed amenable to policy intervention: they depend on the level of abatement effort, the structure of technology, and the parameters of preferences, and they can be relaxed by sufficiently ambitious environmental policy. But one condition stands apart. The requirement that $f'(P) + \rho > 0$ at any candidate steady state, where $f$ is the ecosystem's pollution absorption function and $\rho$ is the social discount rate, is not a feasibility condition. It is an ecological condition, determined by the biophysical properties of the natural system and its interaction with the planner's time horizon. When this condition fails, no saddlepoint equilibrium exists, and no policy instrument operating through the economic margin can restore one. The economy is irreversibly locked out of sustainability, not by a failure of policy, but by a structural property of the ecosystem. Abatement can shift the position of steady states and expand the set of feasible sustainable paths, but it cannot substitute for ecological regenerative capacity. This distinction, between what economics can fix and what it cannot, is one of the most consequential insights that the theory of optimal growth with pollution can offer to the policy debate.

The condition $f'(P) + \rho > 0$ derived in this paper has no direct antecedent in the existing literature on pollution and optimal growth, for a simple structural reason: it cannot arise in models without capital accumulation. In the canonical pollution-only frameworks of Forster (1975) and Brock and Starrett (2003), the state space is one-dimensional and stability is determined by the sign of $-f'(P)$ alone, a purely ecological condition. The social discount rate plays no role in the stability analysis. It is only when pollution dynamics are embedded in a full Ramsey growth model, with endogenous capital accumulation and an optimally chosen consumption path, that the planner's time preference $\rho$ enters the stability condition alongside the curvature of the absorption function. The resulting condition is neither purely ecological nor purely economic: it is a hybrid, reflecting the interaction between the regenerative capacity of the ecosystem and the impatience of the social planner. When this condition fails, the economy is structurally locked out of sustainability in a sense that has no analogue in simpler frameworks, and that no policy instrument operating through the economic margin can remedy.

The analytical foundations of our approach draw on two distinct bodies of work. The characterisation of optimal paths under non-concavity in capital accumulation was established by Nishimura and Dechert (1983), who demonstrated that non-concave production functions generate multiple steady states and render optimal trajectories history-dependent. In the environmental domain, Brock and Starrett (2003) showed that analogous non-concavities in pollution dynamics, arising from hump-shaped natural absorption functions, produce irreversibility and bistability in ecosystem management problems. Our contribution is to embed the non-convex pollution structure of Brock and Starrett (2003) within a Ramsey growth framework of the type analysed by Nishimura and Dechert (1983), thereby generating a condition for sustainability, $f'(P) + \rho > 0$, that is absent from either framework considered in isolation. This said, most of the mathematical difficulties related to non-concavity are brilliantly treated in Brock and Starett (2003), we insist more in this paper on the conceptual and economic implications of embedding the Forster model into the Ramsey set-up.

The remainder of the paper is organised as follows. 
Section 2 reviews the benchmark uncontrolled model and establishes the basic irreversibility threshold. 
Section 3 introduces the controlled model with abatement as a control, derives the reduced dynamical system, characterizes the steady states and proves the main irreversibility theorem. Section 4 establishes the preeminence of the ecological-economic condition over the traditional economic feasability constraints and clarifies the role of optimal abatement in ruling out unsustainable pollution paths. It also demonstrates the robustness of the ecological-economic condition to nonlinear specification of abatement under very mild conditions.
Section 5 concludes. 
All technical proofs are gathered in the Appendix.

\section{Benchmark analysis: the Model with Uncontrolled Pollution}
\label{sec:baseline}

\subsection{The Ramsey problem}

As explained in the Introduction, we start exploring the properties of a Ramsey problem (as in the Barro and Sala-i-Martin textbook, 2003, Chapter~2) when irreversible pollution is accounted for.
We first ignore pollution control instruments to have an accurate idea about the impact of irreversible pollution on the shape of optimal trajectories without pollution control.
We shall add such an instrument, namely abatement, in Section~\ref{sec:abatement}.

Let us consider the following Ramsey model subject to pollution:
\begin{subequations}
\begin{align}
\max_{I(t), C(t)} \quad & \int_0^\infty e^{-\rho t} \left[ \ln C(t) - \eta \frac{P(t)^{1+\mu}}{1+\mu} \right] dt, \label{eq:obj_uncontrolled} \\
\text{s.t.} \quad & A K(t)^\alpha = I(t) + C(t), \label{eq:resource_uncontrolled} \\
& \dot K(t) = I(t) - \delta K(t), \quad K(0)=K_0, \label{eq:kapital_uncontrolled} \\
& \dot P(t) = \gamma A K(t)^\alpha - f(P(t)), \quad P(0)=P_0, \label{eq:pollution_uncontrolled}
\end{align}
\end{subequations}
with \(I(t) \ge 0\) and \(C(t) \ge 0\).
Here \(\rho>0\) is the rate of time preference, \(A>0\) and \(0<\alpha<1\) are the parameters of the Cobb-Douglas production function \(Y=A K^\alpha\), and \(\delta\ge 0\) is the depreciation rate of physical capital.
Environmental quality is characterized by the pollution stock \(P(t)\).
Equation \eqref{eq:resource_uncontrolled} is the resource constraint, while \eqref{eq:kapital_uncontrolled} and \eqref{eq:pollution_uncontrolled} govern the accumulation of capital and pollution, respectively.
The utility function depends on consumption \(C\) and pollution \(P\).
In the pollution law of motion \eqref{eq:pollution_uncontrolled}, emissions are proportional to production, which is standard.
Much less standard, the function \(f(P(t))\) represents pollution decay; the decay rate is not exponential as in the standard pollution models but depends on the level of pollution.
This is the key specification that generates irreversible pollution.

\paragraph{Pollution.}
The choice of a law of motion for environmental pollution \(P(t)\) is critical for the pollution irreversibility analysis (Toman and Withagen, 2000; Jones and Manuelli, 2001).
Because of our interest in pollution irreversibility, we assume that pollution is accumulated as a stock (Stokey, 1998; Hritonenko and Yatsenko, 2013).
Moreover, in line with Tahvonen and Withagen (1996), we assume that pollution dynamics are described by \eqref{eq:pollution_uncontrolled}, where the emission factor \(\gamma>0\) reflects the environmental dirtiness of the economy, and \(f(P)\) measures the natural pollution absorption capacity of the environment.
The latter is key in analysing the pollution irreversibility problem.
Following Tahvonen and Withagen (1996), we set the following key assumption.
\begin{assumption} \label{assum_f}
Assume that \(f \in C^2[0,\infty)\) is such that
\begin{equation*}
f(0)=0, \quad f(P)>0 \text{ and } f''(P)<0 \text{ for } P \in (0, \bar P), \quad f(P)=0 \; \text{ for } \; P \ge \bar P, 
\end{equation*}
for some \(\bar P > 0\).
\end{assumption}
By Assumption \ref{assum_f}, \(f(P)\) has a finite maximum over \((0,\bar P)\):
\[
\bar f := \max_{P \ge 0} f(P) > 0.
\]
For clarity, we assume that the maximum is reached at a single point \(P_{\max} \in (0, \bar P)\), see Figure~1 just below.

\begin{figure}[t!]
\begin{center}
\vspace{-1cm}
\includegraphics[width=0.7 \textwidth]{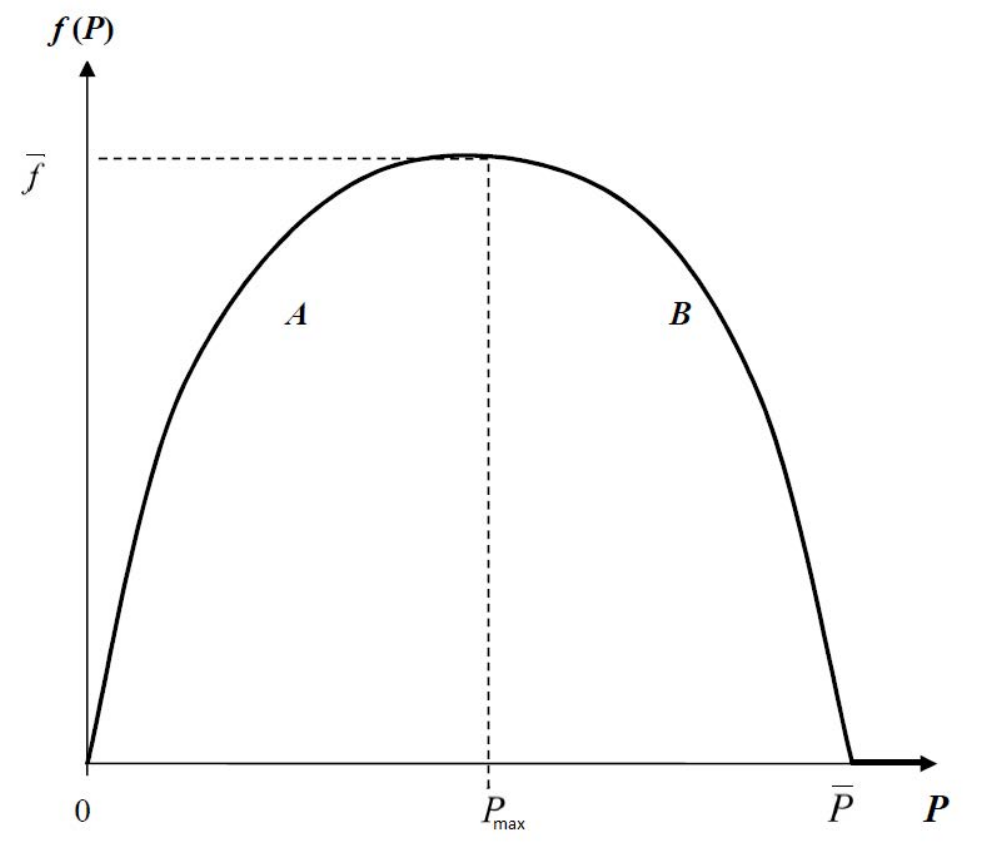}
\vspace{-1.3cm}
\end{center}
\medskip
\caption{The shape of the decay function, f(P)\\ \vspace{.1cm}}
\label{fig:functionf}
\end{figure}

\begin{remark}
The condition \(\bar f > \gamma\) is necessary for the environmental sustainability of the economy \eqref{eq:obj_uncontrolled}--\eqref{eq:pollution_uncontrolled}.
It requires the maximal value \(\bar f\) of \(f(P)\) to be larger than the emission factor \(\gamma\).
Otherwise, the pollution equation \eqref{eq:pollution_uncontrolled} leads to \(P'(t) \ge \gamma - \bar f > 0\) for all \(t\), i.e. pollution \(P(t)\) increases at any economic activity and eventually reaches the pollution irreversibility threshold.
\end{remark}

\paragraph{Utility.}
In line with the economic-environmental literature (Gradus and Smulders, 1993; Stokey, 1998; Byrne, 1997; Economides and Philippopoulos, 2008; Chen et al., 2009; Bréchet et al., 2013; Hritonenko and Yatsenko, 2013, among others), we consider the following additively separable utility function:
\begin{equation}
U(C,P) = \ln C - \eta \frac{P^{1+\mu}}{1+\mu},
\label{eq:utility}
\end{equation}
where the parameter \(\eta>0\) is the environmental vulnerability factor, while \(\mu\) reflects the marginal disutility of pollution, increasing at \(\mu>0\) and decreasing at \(\mu<0\).
The model \eqref{eq:obj_uncontrolled}--\eqref{eq:utility} incorporates key ingredients of the pollution irreversibility problem within a Ramsey setup and without pollution control.

\newpage

\subsection{Preliminary characterization of optimal paths}

The problem contains three equality constraints \eqref{eq:resource_uncontrolled}--\eqref{eq:pollution_uncontrolled}, which leaves us with one decision variable \(I\) and three state variables \(K\), \(C\), and \(P\).
The present-value Hamiltonian for the problem \eqref{eq:obj_uncontrolled}--\eqref{eq:pollution_uncontrolled} is
\begin{equation}
H = \ln C - \eta \frac{P^{1+\mu}}{1+\mu}
+ \lambda_1 (A K^\alpha - I - C)
+ \lambda_2 (I - \delta K)
+ \lambda_3 (\gamma A K^\alpha - f(P)),
\label{eq:hamiltonian_uncontrolled}
\end{equation}
where the dual variables \(\lambda_1, \lambda_2, \lambda_3\) are associated with equalities \eqref{eq:resource_uncontrolled}--\eqref{eq:pollution_uncontrolled}, and \(\mu_1\) is related to the constraint \(I \ge 0\).
The first-order extremum condition for the decision variable \(I\) is
\begin{equation}
-\lambda_1 + \lambda_2 = 0,
\label{eq:foc_I_uncontrolled}
\end{equation}
or \(\lambda_1 = \lambda_2\) in the case of \(I>0\).
The first-order conditions for the state variables \(K\) and \(P\) are, respectively:
\begin{align}
\dot \lambda_1 &= \rho \lambda_1 - \lambda_1 \alpha A K^{\alpha-1} - \lambda_3 \gamma \alpha A K^{\alpha-1}, \label{eq:costate_K_uncontrolled} \\
\dot \lambda_3 &= \rho \lambda_3 + \eta P^\mu + \lambda_3 f'(P), \label{eq:costate_P_uncontrolled}
\end{align}
and the transversality conditions take the form \(\lim_{t\to\infty} e^{-\rho t} \lambda_i(t) K(t) = 0\).
Straightforward substitutions in the system \eqref{eq:costate_K_uncontrolled}--\eqref{eq:costate_P_uncontrolled}, combined with \eqref{eq:resource_uncontrolled}--\eqref{eq:pollution_uncontrolled}, allow us to write the following nonlinear ODE system for interior optimal trajectories \(K(t), C(t), P(t)\), \(0 \le t < \infty\):
\begin{subequations}
\begin{align}
\dot K &= A K^\alpha - C - \delta K, \label{eq:ode_K_uncontrolled} \\
\dot C &= C \left( \alpha A K^{\alpha-1} - \rho - \delta \right), \label{eq:ode_C_uncontrolled} \\
\dot P &= \gamma A K^\alpha - f(P). \label{eq:ode_P_uncontrolled}
\end{align}
\end{subequations}

We start the analysis of the dynamics by looking at the possible stationary equilibria, if any:
\[
K(t) = \bar K, \quad C(t) = \bar C, \quad P(t) = \bar P,
\]
where \(\bar K, \bar C, \bar P\) are positive constants.

The explicit formulas for stationary solutions of \eqref{eq:ode_K_uncontrolled}--\eqref{eq:ode_P_uncontrolled} are
\begin{align}
\bar K &= \left( \frac{\alpha A}{\rho + \delta} \right)^{\frac{1}{1-\alpha}}, \label{eq:barK_uncontrolled} \\
\bar C &= A \bar K^\alpha - \delta \bar K, \label{eq:barC_uncontrolled} \\
f(\bar P) &= \gamma A \bar K^\alpha, \label{eq:barP_uncontrolled}
\end{align}
where the function \(f^{-1}(z)\), \(0 \le z \le \bar f\), is the inverse of \(z = f(P)\), \(0 \le P < \infty\).

This nonlinear system can have multiple solutions under Assumption \ref{assum_f} on \(f(P)\).
Indeed, the function \(z=f(P)\), \(0 \le P \le \bar P\), has two inverse functions:
\begin{align}
f_1^{-1}(z) &:= \text{the inverse on } (0, P_{\max}), \label{eq:inv1} \\
f_2^{-1}(z) &:= \text{the inverse on } (P_{\max}, \bar P), \label{eq:inv2}
\end{align}
with \(0 \le z \le \bar f\).
These inverse functions correspond to the arcs \(A\) and \(B\) of the curve \(z=f(P)\) from Figure~1. See Figure~2 below.
\begin{figure}[t!]
\begin{center}
\vspace{-1cm}
\includegraphics[width=0.7 \textwidth]{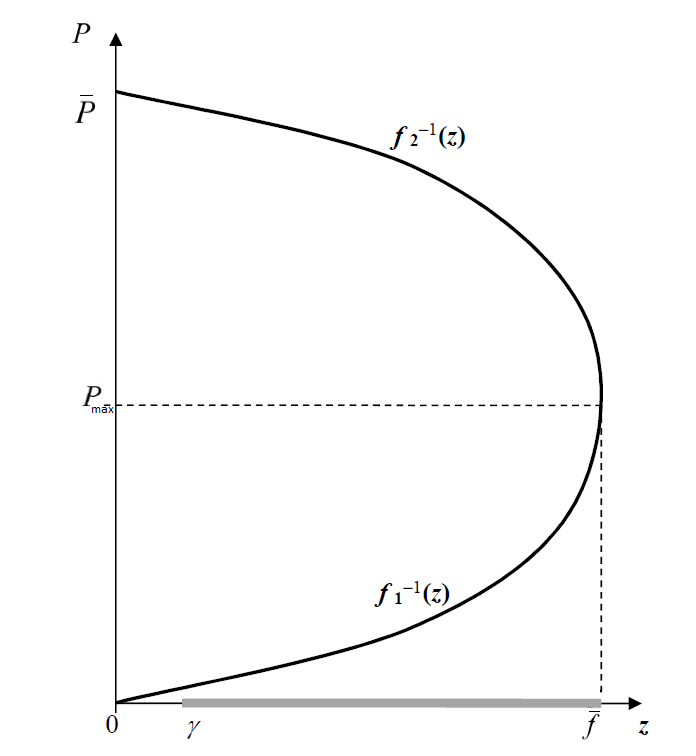}
\vspace{-1.3cm}
\end{center}
\medskip
\caption{The shape of the inverse of the decay function, f(P)\\ \vspace{.1cm}}
\label{fig:invfunctionf}
\end{figure}

If \(\gamma A \bar K^\alpha < \bar f\), then there are two values \(0 < \bar P_1 < P_{\max} < \bar P_2 < \bar P\) that satisfy equation \eqref{eq:barP_uncontrolled}.
Accordingly, the system \eqref{eq:barK_uncontrolled}--\eqref{eq:barP_uncontrolled} has two solutions \((\bar K, \bar C, \bar P_1)\) and \((\bar K, \bar C, \bar P_2)\), which represent two steady states of the problem.
Their convergence/stability properties are investigated in a standard way using the indirect Lyapunov method (see Appendix \ref{appa}). A key feature of the local stability analysis, as demonstrated in Appendix \ref{appa}, is that saddlepoint stability does depend on the sign of $f'(\bar P_i)$, $i=1,2$. Indeed, the linearized three-dimensional system displays 3 eigenvalues, one being exactly equal to $-f'(\bar P_i)$ and the two others have necessarily opposite signs. Because the dynamic system has only one non-predetermined variable, that is consumption, the Blanchard-Kahn conditions require $-f'(\bar P_i)<0$ or $f'(\bar P_i)>0$. Henceforth, \textbf{the steady state equilibrium cannot be saddlepoint stable if it is on the descending branch of the Forster decay curve}, which makes perfect economic sense.
 
Summing up, the structure of the state equations \eqref{eq:resource_uncontrolled}--\eqref{eq:pollution_uncontrolled} allows us to obtain a clear stability result, which we summarise in the following theorem.

\begin{theorem}
\label{thm:uncontrolled_steady}
The optimization problem \eqref{eq:obj_uncontrolled}--\eqref{eq:pollution_uncontrolled} has at most two steady states.
Let \(f(P_0) < \bar f\), where \(\bar P = \sup\{P : f(P)>0\}\).
Then, three cases are possible:
\begin{enumerate}
\item If \(\gamma A \bar K^\alpha < \bar f\), then the problem has one asymptotically stable steady state \((\bar K, \bar C, \bar P_1)\) with \(\bar P_1 < P_{\max}\), and one unstable steady state \((\bar K, \bar C, \bar P_2)\) with \(\bar P_2 > P_{\max}\).
\item If \(\gamma A \bar K^\alpha = \bar f\), there is one unstable steady state \((\bar K, \bar C, \bar P_{\max})\).
\item If \(\gamma A \bar K^\alpha > \bar f\), then the problem has no steady states.
\end{enumerate}
\end{theorem}

Two reversible steady states are possible in the problem \eqref{eq:obj_uncontrolled}--\eqref{eq:pollution_uncontrolled} with an inverted U-shaped \(f(P)\), but only one, \((\bar K, \bar C, \bar P_1)\), is asymptotically stable and, so, possesses an economic value.
The capital \(\bar K\) and consumption \(\bar C\) are the same in both equilibrium points, while the pollution is smaller, \(\bar P_1 < \bar P_2\), in the stable point.
We shall develop these technical points in more detail through an explicit exercise (with a fully specified \(f(P)\)) below.

At this point, note that there exists a threshold for environmental reversibility (that is, for the absorption capacity to remain positive) in the Ramsey economy with uncontrolled pollution, given by the inequality
\begin{equation}
\gamma A \bar K^\alpha > \bar f,
\qquad \text{where } \bar K = \left( \frac{\alpha A}{\rho+\delta} \right)^{\frac{1}{1-\alpha}}.
\label{eq:irrev_threshold_uncontrolled}
\end{equation}

If \(\bar f\) is large enough so that \eqref{eq:irrev_threshold_uncontrolled} does \emph{not} hold, then the economy is likely to avoid falling into environmental irreversibility, and pollution will converge to a finite value in the long term.
In contrast, irreversibility is granted if the pollution absorption capacity is bounded to be low and \eqref{eq:irrev_threshold_uncontrolled} holds.
Condition \eqref{eq:irrev_threshold_uncontrolled} tells us much more about the likelihood of irreversible pollution.
First, the larger \(A\) or \(\gamma\), the more difficult it is to avoid irreversible pollution (as this requires, by \eqref{eq:irrev_threshold_uncontrolled}, absorption capacities to be larger).
This is perfectly natural, as larger \(A\) or \(\gamma\) raises the level of pollution, indirectly via the stock of capital for \(A\), and directly for \(\gamma\).
A similar outcome emerges for the capital share \(\alpha\), for the same reason.
As the discount rate \(\delta\) or the capital depreciation rate \(\rho\) drops, the probability of avoiding irreversibility increases (as the right-hand side of \eqref{eq:irrev_threshold_uncontrolled} diminishes).
Here again, the mechanism operates through capital accumulation: a larger discount or lower capital depreciation rate leads to slower capital accumulation, thereby putting a brake on production and pollution.

This should not come as a surprise: in the model without pollution control, a dichotomy operates between the pure Ramsey part of the model, namely the capital accumulation equation \eqref{eq:ode_K_uncontrolled} and the Keynes-Ramsey rule \eqref{eq:ode_C_uncontrolled}, both independent of pollution, and the pollution outcome, which is computed via the pollution law of motion \eqref{eq:ode_P_uncontrolled} for a given capital stock.
This has implications for the steady state (as reflected in equations \eqref{eq:barK_uncontrolled}--\eqref{eq:barP_uncontrolled}), for the comparative statics shown above, and also for the dynamics and stability properties as shown in the subsection below.

Also note that the irreversibility criterion \eqref{eq:irrev_threshold_uncontrolled} depends only on the maximal value \(\bar f\) of the pollution absorption capacity \(f(P)\), but the resulting stationary pollution levels \(\bar P_1\) and \(\bar P_2\) obviously depend on the shape of \(f(P)\).
To analyse this dependence in more detail, we consider a quadratic \(f(P)\) in the next subsection.
In their explicit calculations, Tahvonen and Withagen (1996) consider linear functions \(f(P)\), while the setting adopted in El Ouardighi et al. (2014) is fully linear-quadratic.
The exercise performed below is therefore totally original; it highlights readily and precisely some new properties sketched briefly above.
It will in particular allow us to neatly discuss the optimality of reversible versus irreversible pollution paths.

\subsection{Optimal reversible vs. irreversible pollution paths}

We shall specify \(f(P)\) as follows.

\begin{corollary}[Case of quadratic pollution absorption capacity]
\label{cor:quadratic_uncontrolled}
Let
\begin{equation}
f(P) = a P - b P^2, \quad \text{with } a>0,\ b>0.
\label{eq:f_quadratic}
\end{equation}
Then the steady-state pollution levels are
\begin{align}
\bar P_1 &= \frac{a - \sqrt{a^2 - 4b \gamma A \bar K^\alpha}}{2b}, \label{eq:P1_quadratic} \\
\bar P_2 &= \frac{a + \sqrt{a^2 - 4b \gamma A \bar K^\alpha}}{2b}. \label{eq:P2_quadratic}
\end{align}
Specifically, if \(\gamma A \bar K^\alpha < \bar f = a^2/(4b)\), then the stable steady-state pollution is \(\bar P_1\), given by \eqref{eq:P1_quadratic}.
\end{corollary}

By Corollary~\ref{cor:quadratic_uncontrolled}, the asymptotically stable steady-state pollution \(\bar P_1\) in Theorem~\ref{thm:uncontrolled_steady} depends on the maximal value \(\bar f\) of the absorption capacity \(f(P)\) and on the pollution level \(P_{\max} = a/(2b)\) at which it occurs.
Specifically, \(\bar P_1\) is linearly related to \(P_{\max}\) and \(\bar f\).
Therefore, the stable sustainable pollution level \(\bar P_1\) is lower when \(P_{\max}\) is smaller (that is, when the shape of the quadratic \(f(P)\) is steeper).
In the case of a sufficiently large \(\bar f\), the steady-state pollution \(\bar P_1\) is approximately proportional to the steady-state capital \(\bar K\).

As alluded to in the analysis above, the formulas \eqref{eq:barK_uncontrolled} and \eqref{eq:barC_uncontrolled} for the steady-state levels coincide with formulas (1.27) in Barro and Sala-i-Martin (2003).
The pair \((\bar K, \bar C)\) describes the so-called \emph{golden path} in the classic Ramsey model \eqref{eq:obj_uncontrolled}--\eqref{eq:kapital_uncontrolled} with \(P \equiv 0\) and does not depend on the pollution characteristics \(\gamma\) and \(f(P)\).
This is natural, since the economy in the model \eqref{eq:obj_uncontrolled}--\eqref{eq:pollution_uncontrolled} spends nothing on environmental cleanup and simply accepts the given level of accumulated pollution.

Given this neat dichotomy, the dynamic analysis of the model derives to a large extent from that of the underlying Ramsey sub-model, completed by a straightforward (local) analysis of the pollution law of motion \eqref{eq:ode_P_uncontrolled}.
A detailed qualitative analysis of optimal trajectories \(K(t)\) and \(C(t)\) in the classic Ramsey model without pollution was provided in Barro and Sala-i-Martin (2003, Chapter~2).
Extending their results to our model, we obtain the following picture of transition dynamics for the optimization problem \eqref{eq:obj_uncontrolled}--\eqref{eq:pollution_uncontrolled}.

\begin{theorem}
\label{thm:uncontrolled_dynamics}
Let \(\rho > \gamma (1+\mu)\).
Then there exists a unique optimal trajectory \((K^*, C^*, P^*)\) such that:
\begin{itemize}
\item the optimal capital \(K^*(t)\) and consumption \(C^*(t)\) monotonically converge to the steady-state levels \(\bar K\) and \(\bar C\);
\item the optimal pollution \(P^*(t)\) converges to \(\bar P_1\) when \(P(0) \le \bar P_1\)  and \eqref{eq:irrev_threshold_uncontrolled} holds, and \(P^*(t) \to \infty\) as \(t \to \infty\) when \(P(0) > \bar P_1\) or \eqref{eq:irrev_threshold_uncontrolled} fails.
\end{itemize}
\end{theorem}

Theorem~\ref{thm:uncontrolled_dynamics} shows that the optimal capital \(K(t)\) and consumption \(C(t)\) in the third case of Theorem~\ref{thm:uncontrolled_steady} still approach the stationary values \((\bar K, \bar C)\), while the pollution \(P(t)\) increases indefinitely.
Two points are worth making at this stage.
First, Theorem~\ref{thm:uncontrolled_dynamics} is a natural implication of the dichotomy property already highlighted above.
The capital and consumption dynamics (and steady states) are driven by the Ramsey sub-model, while pollution dynamics derive from the law of motion \eqref{eq:ode_P_uncontrolled} for a given (optimal) capital path.
By \eqref{eq:ode_P_uncontrolled}, the evolution of the pollution stock depends on the difference between emissions, \(\gamma A K^\alpha\), which depend entirely on the capital dynamics derived in the Ramsey sub-model, and the amount of pollution absorbed, \(f(P)\).
If \eqref{eq:irrev_threshold_uncontrolled} holds, the former term dominates, and pollution grows without bound.
Otherwise, the latter term dominates, and pollution converges to a steady state.

Second, and more interestingly, when the irreversibility threshold is crossed, i.e. when \eqref{eq:irrev_threshold_uncontrolled} holds, pollution grows only linearly in the long run: this comes directly from the pollution law of motion \eqref{eq:ode_P_uncontrolled} as \(K\) converges to \(\bar K\) and \(f(P)\) is bounded by Assumption \ref{assum_f}.
Henceforth, the objective functional \eqref{eq:obj_uncontrolled} remains bounded along the explosive paths.\footnote{More precisely, when (13) holds, $P(t)$ grows linearly since $K(t) \to \bar{K}$ and $f(P)$ is bounded by Assumption 1. The disutility integrand thus behaves as $e^{-\rho t} P(t)^{1+\mu} \sim e^{-\rho t} t^{1+\mu}$, which is integrable for any $\rho > 0$. The transversality condition $\lim_{t\to\infty} e^{-\rho t} \lambda_3(t)P(t) = 0$ is also satisfied: since $\lambda_3(t)$ grows at most polynomially (driven by the $\eta P^\mu$ forcing in the costate equation), the product $e^{-\rho t}\lambda_3(t)P(t)$ decays exponentially. This confirms that irreversible paths are not only feasible but genuinely optimal in the sense of the Pontryagin conditions.}
This means that when \eqref{eq:irrev_threshold_uncontrolled} holds, the irreversible pollution regime sets in and is \emph{optimal}.

\section{Optimal Growth with Pollution Abatement}
\label{sec:abatement}

\subsection{The Controlled Problem}
\label{subsec:controlled_problem}

We now extend the model of Section \ref{sec:baseline} by introducing a pollution 
abatement policy. The representative agent chooses consumption $C(t)$, investment 
$I(t)$, and an abatement effort $\theta(t) \in [0,1]$ to maximise the discounted 
utility stream:
\begin{equation}
\max_{C,\theta} \int_0^\infty e^{-\rho t} \left[ \ln C - \eta 
\frac{P^{1+\mu}}{1+\mu} \right] dt
\label{eq:obj_controlled}
\end{equation}
subject to the capital accumulation equation:
\begin{equation}
\dot{K} = (1-\theta)AK^\alpha - C - \delta K, \qquad K(0) = K_0 > 0,
\label{eq:kapital_controlled}
\end{equation}
the pollution dynamics:
\begin{equation}
\dot{P} = \gamma(1-\theta)AK^\alpha - f(P), \qquad P(0) = P_0 > 0,
\label{eq:pollution_controlled}
\end{equation}
and the non-negativity constraints $C(t) \geq 0$, $\theta(t) \in [0,1]$.

The abatement effort $\theta(t)$ represents the fraction of gross output $AK^\alpha$ 
diverted from productive use toward pollution control. A fraction $\theta$ of output 
is sacrificed, reducing both the resources available for consumption and investment 
and the flow of emissions proportionally:
\begin{itemize}
\item net output available for consumption and investment is $(1-\theta)AK^\alpha$;
\item gross emissions are reduced from $\gamma AK^\alpha$ to 
$\gamma(1-\theta)AK^\alpha$.
\end{itemize}
There are several other modelling choices for the emission function $E\left(Y(t), \theta(t)\right)$, increasing with output and decreasing with an abatement indicator. Smulders and Gradus (1996) considered indeed a more complex specification for the emission flows as $Y^{\pi} \; \theta^{-\lambda}$, with $\lambda>\pi$ while Br\'echet et al. (2013) specialize in the 0-homogenous case $\lambda=\pi=1$. Needless to say, the ``correct” specification depends on the type of pollutant and abatement technology used. Our assumption involves less algebraic complexity, and will be instrumental in deriving clear-cut results for the local stability properties required. It is deliberately parsimonious: a single control $\theta$ governs 
both the output cost and the emission reduction, ensuring clean analytical results 
while preserving the economic intuition of a genuine trade-off between growth and 
environmental quality. we shall show later that the main result of the paper (the ecological-economic sustainability condition enhanced in the Introduction) is robust to nonlinear specifications of abatement under very mild conditions.

\subsection{Optimality Conditions}
\label{subsec:optimality}

We now display the necessary optimality conditions for an interior optimum. Of course, because the utility function is logarithmic in consumption, the corner solution $\theta=1$ cannot be optimal. However, $\theta=0$ can be optimal, it corresponds to zero abatement, which takes us back to the benchmark model studied in Section~2. That is precisely why we focus on the interior optima if any.

The current-value Hamiltonian associated with problem 
(\ref{eq:obj_controlled})--(\ref{eq:pollution_controlled}) is:
\begin{equation}
\mathcal{H} = \ln C - \eta \frac{P^{1+\mu}}{1+\mu} 
+ \lambda_1 \left[(1-\theta)AK^\alpha - C - \delta K\right] 
+ \lambda_3 \left[\gamma(1-\theta)AK^\alpha - f(P)\right]
\label{eq:hamiltonian_controlled}
\end{equation}
where $\lambda_1(t)$ and $\lambda_3(t)$ are the costate variables associated with 
capital $K$ and pollution $P$ respectively. We develop below the necessary optimality conditions, the transversality conditions for $K$ and $P$ are standard as in the benchmark.\footnote{As in the benchmark case, the boundedness of the objective function along the irreversible pollution paths can be directly checked.}

\medskip
\noindent\textit{Optimality with respect to $C$.} The first-order condition 
$\partial\mathcal{H}/\partial C = 0$ gives:
\begin{equation}
\frac{1}{C} = \lambda_1.
\label{eq:foc_C_controlled}
\end{equation}
\noindent\textit{Optimality with respect to $\theta$.} The first-order condition 
$\partial\mathcal{H}/\partial\theta = 0$ gives:
\begin{equation}
-\lambda_1 AK^\alpha - \lambda_3 \gamma AK^\alpha = 0.
\label{eq:foc_theta}
\end{equation}
Since $AK^\alpha > 0$, condition (\ref{eq:foc_theta}) simplifies to:
\begin{equation}
\lambda_1 + \gamma\lambda_3 = 0 
\qquad \Longrightarrow \qquad 
\lambda_3 = -\frac{\lambda_1}{\gamma}.
\label{eq:lambda_relation}
\end{equation}
This is a static optimality condition: at every instant, the marginal cost of 
abatement in terms of foregone output (captured by $\lambda_1$) must equal its 
marginal benefit in terms of reduced pollution damage (captured by $\gamma\lambda_3$). 
Equation (\ref{eq:lambda_relation}) links the two costate variables at all times 
along the optimal path.

\medskip
\noindent\textit{Costate equations.} The costate equation for capital is:
\begin{equation}
\dot{\lambda}_1 = \rho\lambda_1 - \frac{\partial\mathcal{H}}{\partial K} 
= \lambda_1\left[\rho + \delta - (1-\theta)\alpha AK^{\alpha-1}\right].
\label{eq:costate_lambda1_simplified}
\end{equation}
The costate equation for pollution is:
\begin{equation}
\dot{\lambda}_3 = \rho\lambda_3 - \frac{\partial\mathcal{H}}{\partial P} 
= \rho\lambda_3 + \eta P^\mu + \lambda_3 f'(P).
\label{eq:costate_lambda3}
\end{equation}
Differentiating (\ref{eq:foc_C_controlled}) with respect to time yields 
$\dot{\lambda}_1/\lambda_1 = -\dot{C}/C$. Substituting into 
(\ref{eq:costate_lambda1_simplified}) gives the \textit{Keynes-Ramsey rule}:
\begin{equation}
\frac{\dot{C}}{C} = (1-\theta)\alpha AK^{\alpha-1} - \rho - \delta.
\label{eq:KR_controlled}
\end{equation}
Comparing with the uncontrolled Keynes-Ramsey rule of Section \ref{sec:baseline}, 
the key modification is the appearance of $(1-\theta)$ in the net return to capital. 
Abatement reduces the effective return on capital by diverting a fraction $\theta$ of 
output away from productive accumulation, thereby slowing consumption growth relative 
to the uncontrolled economy.

\medskip
\noindent Differentiating (\ref{eq:lambda_relation}) with respect to time and 
substituting (\ref{eq:costate_lambda1_simplified}) and (\ref{eq:costate_lambda3}) 
yields a \textit{pollution pricing rule} that must hold at all times along the 
optimal path:
\begin{equation}
\eta P^\mu = \frac{\lambda_1}{\gamma}
\left[\rho + \delta + f'(P) - (1-\theta)\alpha AK^{\alpha-1}\right].
\label{eq:pricing_rule}
\end{equation}
Condition (\ref{eq:pricing_rule}) equates the marginal disutility of pollution 
$\eta P^\mu$ to the shadow cost of emissions, adjusted for the net return on capital 
and the absorption capacity of the environment. The implicit price of pollution along 
the optimal path is thus pinned down by the interplay between preferences, technology, 
and environmental dynamics.

\begin{remark} 
The necessary conditions derived in this section are not complemented by sufficient conditions for a global optimum. The non-concavity of the absorption function $f(P)$, which is the structural feature generating pollution irreversibility, implies that the current-value Hamiltonian (20) is not jointly concave in the state and control variables. Standard sufficiency theorems (Mangasarian or Arrow-Kurz) therefore do not apply globally. On the ascending branch of $f(P)$, where $f''(P) < 0$ locally, a restricted sufficiency argument can be made in the neighbourhood of the stable steady state; the saddlepoint structure established in Theorem 4 and the uniqueness of the convergent trajectory on the stable manifold together provide strong evidence that the identified interior path is indeed optimal within its basin of attraction. A global sufficiency result would require either a convexification of the problem or a verification approach via the Hamilton-Jacobi-Bellman equation, both of which are beyond the scope of the present analysis. 
\end{remark}

\subsection{Steady-State Analysis}
\label{subsec:steady_state_controlled}

At a steady state $(\bar{K},\bar{C},\bar{P},\bar{\theta})$, all time derivatives 
vanish. From $\dot{C}/C = 0$ in (\ref{eq:KR_controlled}):
\begin{equation}
(1-\bar{\theta})\,\alpha A\bar{K}^{\alpha-1} = \rho + \delta.
\label{eq:ss_KR_controlled}
\end{equation}
From $\dot{K} = 0$ in (\ref{eq:kapital_controlled}):
\begin{equation}
(1-\bar{\theta})\,A\bar{K}^\alpha = \bar{C} + \delta\bar{K}.
\label{eq:ss_K_controlled}
\end{equation}
From $\dot{P} = 0$ in (\ref{eq:pollution_controlled}):
\begin{equation}
\gamma(1-\bar{\theta})\,A\bar{K}^\alpha = f(\bar{P}).
\label{eq:ss_P_controlled}
\end{equation}
Condition (\ref{eq:ss_KR_controlled}) implicitly defines the steady-state capital 
stock as a function of $\bar{\theta}$:
\begin{equation}
\bar{K}(\bar{\theta}) = \left(\frac{\alpha A(1-\bar{\theta})}{\rho+\delta}
\right)^{\frac{1}{1-\alpha}}.
\label{eq:barK_controlled}
\end{equation}
Comparing with the uncontrolled steady-state capital $\bar{K} = 
\left(\alpha A/(\rho+\delta)\right)^{1/(1-\alpha)}$ from Section \ref{sec:baseline}:
\begin{equation}
\bar{K}(\bar{\theta}) = (1-\bar{\theta})^{\frac{1}{1-\alpha}}\,\bar{K} < \bar{K}
\qquad \text{for all } \bar{\theta} > 0.
\label{eq:barK_comparison}
\end{equation}
Abatement unambiguously reduces the steady-state capital stock. The fraction of 
output devoted to pollution control lowers the net return on capital, inducing the 
economy to accumulate less capital in the long run.

\paragraph{Breaking the dichotomy} 

In the uncontrolled model of Section \ref{sec:baseline}, the steady-state pair $(\bar{K},\bar{C})$ was determined 
entirely by the Ramsey sub-system, independently of pollution characteristics. This 
dichotomy breaks down in the controlled model. Using $\bar{\lambda}_1 = 1/\bar{C}$ 
and (\ref{eq:ss_KR_controlled}), the pollution pricing rule (\ref{eq:pricing_rule}) 
at steady state reduces to:
\begin{equation}
\eta\bar{P}^\mu = \frac{1}{\gamma\bar{C}}\left[\delta + f'(\bar{P})\right].
\label{eq:ss_pricing}
\end{equation}

Condition (\ref{eq:ss_pricing}) links $\bar{P}$, $\bar{C}$, and the shape of $f$ 
simultaneously. The pollution level, the consumption level, and the abatement rate 
are now jointly determined: the economy can no longer treat pollution as a residual 
outcome of an independently determined capital path. More importantly, equation \ref{eq:ss_pricing} implies that no steady state equilibrium
can exist unless

\[f'(\bar{P}) +\delta>0 .\]

This is the ecological-economic condition explicitly mentioned in the Introduction, we shall refer to it as the \textbf{(EEC)} condition from now on.  The economic interpretation of (EEC) is illuminating. On the ascending branch of $f(P)$, the ecosystem's absorption capacity is increasing in pollution: a marginal increase in $P$ raises the flow of natural remediation, providing a self-correcting force. On the descending branch, this corrective capacity is eroding: higher pollution weakens the ecosystem further, and the absorption function declines. When $|f'(\bar{P})| > \rho$, the rate at which the ecosystem deteriorates at the margin exceeds the social discount rate. The planner, even endowed with perfect foresight and unconstrained in the choice of abatement, cannot respond quickly enough to prevent the ecological dynamics from overwhelming the optimal trajectory. Irreversibility, in this reading, is not a failure of policy design or economic incentives: it is a failure of the ecosystem itself to provide the regenerative capacity that sustainability requires. It is in this precise sense that the ecological condition is preeminent.

We shall show in Section 4 that this condition is preeminent, dominating in particular the economic feasability conditions. We will also demonstrate that the latter properties would hold even under nonlinear specifications of the abatement technology provided very mild conditions.

\subsection{Local stability}
\label{app:stability_controlled}

The full dynamic system consists of four equations governing $(\dot{K}, \dot{P}, \dot{\lambda_1}, \dot{\lambda_3})$. However, under the linear abatement specification, the first-order condition for $\theta$ yields a closed-form relationship between the two costate variables. Specifically, at an interior optimum, equating the marginal cost and marginal benefit of abatement gives equation \ref{eq:lambda_relation}, that is

$$\lambda_1 = - \gamma \lambda_3$$

This condition holds along the entire optimal path, not merely at the steady state. It effectively ties $\lambda_1$ to $\lambda_3$, reducing the dimensionality of the system from four to three. Substituting $\lambda_1 = - \gamma \lambda_3$ throughout, the relevant dynamic system becomes one in $(K, P, \lambda_3)$ only (or equivalently in $(K, P,C)$ as below), and the stability analysis proceeds by evaluating the $3\times 3$ Jacobian of this reduced system at the interior steady state $(K^, P^, \lambda_3)$. This dimension reduction is a direct consequence of the linearity of abatement in the emission equation and would not be available under a general nonlinear abatement specification $g(\theta)$, which is one reason we maintain the linear formulation as our baseline.

The controlled dynamical system then consists of three differential equations 
in $(K, C, P)$. From the optimality conditions of the Hamiltonian , the equations of motion are:

\begin{align}
    \dot{K} &= (1-\bar{\theta})AK^\alpha - C - \delta K, 
    \label{app:Kdot} \\[6pt]
    \dot{C} &= C\left[\alpha(1-\bar{\theta})AK^{\alpha-1} 
               - \delta - \rho\right], 
    \label{app:Cdot} \\[6pt]
    \dot{P} &= \gamma(1-\bar{\theta})AK^\alpha - f(P),
    \label{app:Pdot}
\end{align}

where $\bar{\theta}$ is the optimal abatement rate, treated as 
constant in the neighbourhood of the steady state for the purpose 
of local stability analysis.\footnote{This is without loss of generality: at the steady state, $\bar{\theta}$ is determined by the algebraic optimality condition (33), jointly with $(\bar{K}, \bar{C}, \bar{P})$. Since $\theta$ is a control variable with no independent law of motion, its local deviations from $\bar{\theta}$ are of second order in the deviations of the state-costate vector and do not affect the first-order linearisation. Equivalently, the dimension reduction already incorporates the optimal adjustment of $\theta$ along the path; fixing $\theta = \bar{\theta}$ in the linearised system is therefore consistent with the full optimality conditions.} The Blanchard-Kahn condition for a unique convergent path in a 
three-dimensional system requires exactly one unstable eigenvalue, 
yielding a \emph{two-dimensional} stable manifold, exactly as in the benchmark model with abatement due to the linearity of the abatement specification. Indeed, as stated in Theorem \ref{thm:controlled_steady} below, the mathematical structure of the induced local dynamics are similar, and the algebra as well.

For the latter reason, the central question of this local stability section is more focused on whether optimal abatement can prevent 
irreversible pollution accumulation. Recall from Section \ref{sec:baseline} that 
irreversibility arises in the uncontrolled model whenever:
\begin{equation}
\gamma A\bar{K}^\alpha > \bar{f},
\label{eq:irrev_uncontrolled}
\end{equation}
where $\bar{f} = \max_P f(P)$ denotes the maximum absorption capacity of the 
environment. In the controlled model, the steady-state pollution condition 
(\ref{eq:ss_P_controlled}) requires:
\begin{equation}
\gamma(1-\bar{\theta})\,A\bar{K}(\bar{\theta})^\alpha \leq \bar{f}.
\label{eq:irrev_threshold_controlled}
\end{equation}
Substituting (\ref{eq:barK_controlled}) into (\ref{eq:irrev_threshold_controlled}):
\begin{equation}
\gamma A\bar{K}^\alpha\,(1-\bar{\theta})^{\frac{1}{1-\alpha}} \leq \bar{f}.
\label{eq:irrev_threshold_explicit}
\end{equation}
Since $(1-\bar{\theta})^{1/(1-\alpha)} < 1$ for any $\bar{\theta} > 0$, condition 
(\ref{eq:irrev_threshold_explicit}) is strictly easier to satisfy than 
(\ref{eq:irrev_uncontrolled}). Abatement strictly relaxes the irreversibility 
threshold.

\medskip
\noindent The minimum abatement effort required to avoid irreversibility, the 
\textit{critical abatement rate}, is obtained by solving 
(\ref{eq:irrev_threshold_explicit}) with equality:
\begin{equation}
\theta^* = 1 - \left(\frac{\bar{f}}{\gamma A\bar{K}^\alpha}\right)^{1-\alpha}.
\label{eq:theta_critical}
\end{equation}
Note that $\theta^* > 0$ if and only if $\gamma A\bar{K}^\alpha > \bar{f}$, i.e.\ 
precisely when the uncontrolled economy faces irreversibility. The following theorem 
summarises the structure of steady states in the controlled model.

\begin{theorem}
\label{thm:controlled_steady}
Suppose $f$ is hump-shaped with maximum $\bar{f}$ attained at $P_{\max}$, and let 
$P_0 < P_{\max}$. Then:
\begin{enumerate}
\item[\emph{(i)}] If $\gamma(1-\bar{\theta})\,A\bar{K}(\bar{\theta})^\alpha < 
\bar{f}$, the controlled problem has two steady states: an asymptotically stable 
equilibrium $(\bar{K}(\bar{\theta}),\bar{C}(\bar{\theta}),\bar{P}_1(\bar{\theta}))$ 
with $\bar{P}_1 < P_{\max}$, and an unstable equilibrium with $\bar{P}_2 > P_{\max}$.
\item[\emph{(ii)}] If $\gamma(1-\bar{\theta})\,A\bar{K}(\bar{\theta})^\alpha = 
\bar{f}$, there is a unique, unstable steady state at $P = P_{\max}$.
\item[\emph{(iii)}] If $\gamma(1-\bar{\theta})\,A\bar{K}(\bar{\theta})^\alpha > 
\bar{f}$, no steady state exists and pollution grows without bound.
\end{enumerate}
\end{theorem}

\noindent As alluded to above, the proof follows the same steps as Theorem \ref{thm:uncontrolled_steady} in the benchmark Section 2, we omit it. 

\medskip
\noindent Three cases emerge directly from Theorem \ref{thm:controlled_steady}: if 
$\bar{\theta} > \theta^*$, the economy escapes irreversibility and converges to a 
finite steady-state pollution $\bar{P}_1$; if $\bar{\theta} = \theta^*$, the economy 
sits at the boundary; and if $\bar{\theta} < \theta^*$, irreversible pollution 
persists despite abatement. Abatement therefore softens but does not eliminate the 
irreversibility constraint: when emission intensity $\gamma$ or productivity $A$ is 
sufficiently large, no finite abatement rate $\bar{\theta} \in [0,1)$ can prevent 
unbounded pollution growth.

\section{Optimal abatement and the preeminence of the ecological-economic condition}

The analysis of steady states and their stability reveals a fundamental asymmetry between the conditions governing the sustainability of optimal paths. The (EEC), namely the condition $f'(\bar{P}) + \rho > 0$, occupies a qualitatively distinct position in the hierarchy of constraints. To see why, recall that the existence of a saddlepoint equilibrium requires the Jacobian of the dynamical system to exhibit eigenvalues of opposite sign. As shown in the Appendix, this is equivalent to requiring that the determinant of the Jacobian be strictly negative, which as explained in Section 3 is perfectly consistent with the (EEC). When $f'(\bar{P})<0$, the Jacobian determinant is non-negative, both eigenvalues share the same sign, and no saddlepoint can exist. The optimal path, in this case, does not converge to any interior steady state: the economy is structurally locked out of sustainability, regardless of the level of abatement or the stringency of environmental policy.

This observation deserves emphasis. The feasibility thresholds derived in the preceding analysis, namely the critical abatement rate $\theta$, the capital conditions on $\bar{K}$, the bounds on emission intensity, are all endogenous constraints. They depend on preferences, technology, and policy instruments, and they can in principle be relaxed through appropriate intervention. By contrast, (EEC) is structural: it is determined by the biophysical properties of the ecosystem, summarised in the curvature of the absorption function $f(P)$ at the candidate steady state, and by the social rate of time preference $\rho$. No reallocation of output toward abatement, however ambitious, can alter this condition directly. What abatement can do is shift the steady-state pollution level $\bar{P}$ to a lower value, ideally back onto the ascending branch of $f(P)$, where $f'(\bar{P}) > 0$ and (EEC) holds. The feasibility threshold $\theta$ is precisely the minimum abatement effort required to achieve this. But the logical priority is clear: the ecological condition defines the domain within which sustainability is conceivable; the feasibility conditions determine whether the economy can afford to operate within that domain.

We study just below the robustness of the (EEC) in more nonlinear settings before clarifying the role of abatement policies in the Forster-Romer model.

\subsection{Robustness of the (EEC) to nonlinear abatement technologies}
A critical point to discuss is the robustness of the (EEC) to nonlinear specifications of abatement. We shall easily show that it is robust. Consider the Ramsey model of Section 3 but with the two following modifications:
\begin{equation*}
\dot{K} = AK^\alpha - C - g(\theta) - \delta K, \qquad K(0) = K_0 > 0,
\end{equation*}
and
\begin{equation*}
\dot{P} = E\left(AK^\alpha, \theta\right) - f(P), \qquad P(0) = P_0 > 0,
\end{equation*}
where $\theta$ is the abatement effort, reducing emissions $E(.)$ in the pollution law of motion (that is $\frac{\partial E}{\partial \theta}\equiv E_{\theta}<0$) but increasing the cost of abatement in terms of the final good (that is $g'(\theta)>0$). Writing the current-value Hamiltonian like in Section 3, with $\lambda_1$ the co-state variable for $K$ and $\lambda_3$ the co-state variable for $P$, one can establish easily after a few trivial algebraic operations that
\[\lambda_3= \frac{g'(\theta)}{C E_{\theta}},\]
which is negative as it should be given our general specifications. Now writing the co-state equation for pollution at the steady state (equation (25) in Section 3), one gets:
\[f'(P) + \rho= \frac{-\eta \; P^{\mu}}{\lambda_3}>0 .\]
This prove that our ecological condition is robust to the nonlinearity of abatement in the dynamics of pollution provided the model is correctly specified.  

\subsection{The Effect of Abatement on the Basin of Attraction}

While the (EEC) is clearly a robust condition in the class of the Ramsey-Forster models, the role of abatement should not be neglected. The introduction of optimal abatement has two distinct effects on the 
geometry of the stable manifold and the basin of attraction, which 
together constitute the main message of this subsection.
\begin{itemize}
\item \textbf{ Expansion of the basin of attraction.}
Optimal abatement reduces effective emissions from $\gamma A K^\alpha$ 
to $\gamma(1-\bar{\theta}) A K^\alpha$, lowering the pollution flow 
at every level of capital. The irreversibility threshold \eqref{eq:irrev_uncontrolled} 
is replaced by the controlled condition:
\begin{equation}
    \gamma(1 - \bar{\theta}) A \bar{K}(\bar{\theta})^\alpha > \bar{f},
    \label{eq:irrev_controlled}
\end{equation}
Since $\bar{\theta} > 0$ reduces the left-hand side relative to the 
uncontrolled case, the set of parameter configurations under which 
irreversibility is avoided is strictly larger under optimal abatement. 
Equivalently, for a given parameter configuration, the unstable pollution 
threshold $\bar{P}_2(\bar{\theta})$ shifts outward, expanding the basin of attraction in the pollution 
dimension.

\item \textbf{The ecological conditions remain binding.}
Despite the expansion of the basin of attraction, the structural 
condition for saddlepoint stability:
\begin{equation*}
    f'(\bar{P}_1(\bar{\theta}))> 0,
\end{equation*}
remains a \emph{preeminent constraint} that optimal abatement cannot 
override. This condition depends exclusively on the curvature of the 
natural absorption function $f(P)$ at the steady state and on the 
planner's discount rate $\rho$. It is a structural, biophysical 
requirement: if it fails, no saddlepoint steady state exists on the 
descending branch of $f(P)$, regardless of the level of abatement 
chosen. Abatement shifts the location of the steady state 
$\bar{P}_1(\bar{\theta})$ along the ascending branch of $f(P)$, 
but it cannot move the steady state to a region where the ecological 
condition is violated and simultaneously maintain saddlepoint stability.
\end{itemize}

\smallskip

The key insight from the previous comparison is that abatement is a powerful instrument for 
sustainability since it expands the set of initial conditions from which 
convergence to a clean steady state is possible. However, optimal abatement operates 
within the bounds set by the (EEC), which is essentially determined by nature, 
not by policy.

\subsection{Abatement and the Irreversibility Constraint: A Quadratic Example}
\label{subsec:quadratic_controlled}

To illustrate the results above concretely, we return to the quadratic specification 
$f(P) = aP - bP^2$ introduced in Corollary \ref{cor:quadratic_uncontrolled}. Under this 
specification, $\bar{f} = a^2/(4b)$ and $P_{\max} = a/(2b)$. The steady-state 
pollution levels in the controlled economy are:
\begin{equation}
\bar{P}_{1,2}(\bar{\theta}) = \frac{a \mp \sqrt{a^2 - 
4b\,\gamma(1-\bar{\theta})\,A\bar{K}(\bar{\theta})^\alpha}}{2b},
\label{eq:P12_controlled}
\end{equation}
where $\bar{P}_1$ takes the minus sign and $\bar{P}_2$ the plus sign. Since 
$\gamma(1-\bar{\theta})\,A\bar{K}(\bar{\theta})^\alpha$ is strictly decreasing in 
$\bar{\theta}$, it follows immediately that:
\begin{itemize}
\item $\bar{P}_1(\bar{\theta}) < \bar{P}_1(0)$: abatement lowers the stable 
steady-state pollution level;
\item $\bar{P}_2(\bar{\theta}) > \bar{P}_2(0)$: abatement raises the unstable 
steady-state pollution level;
\item the two steady states move apart as $\bar{\theta}$ increases, widening the 
basin of attraction of the stable equilibrium.
\end{itemize}

\noindent\textbf{Numerical illustration.} To fix ideas, consider the following 
baseline parameterisation:
\begin{center}
\begin{tabular}{llll}
\hline
Parameter & Value & & Interpretation \\
\hline
$A$ & $1.0$ & & Total factor productivity \\
$\alpha$ & $0.33$ & & Capital share \\
$\rho$ & $0.04$ & & Discount rate \\
$\delta$ & $0.05$ & & Depreciation rate \\
$\gamma$ & $0.8$ & & Emission intensity \\
$a$ & $0.6$ & & Linear absorption coefficient \\
$b$ & $0.05$ & & Quadratic absorption coefficient \\
$\eta$ & $0.1$ & & Pollution disutility weight \\
$\mu$ & $1.0$ & & Pollution disutility curvature \\
\hline
\end{tabular}
\end{center}

\medskip
\noindent Under these parameters, the uncontrolled steady-state capital and the 
maximum absorption capacity are:
\[
\bar{K} \approx 6.84, \qquad \bar{f} = \frac{a^2}{4b} = 1.80.
\]
Since $\gamma A\bar{K}^\alpha \approx 1.52 < \bar{f} = 1.80$, the uncontrolled 
economy avoids irreversibility. Now suppose emission intensity rises to $\gamma = 
1.1$, so that $\gamma A\bar{K}^\alpha \approx 2.09 > \bar{f}$: the uncontrolled 
economy faces irreversible pollution. The critical abatement rate 
(\ref{eq:theta_critical}) is:
\[
\theta^* = 1 - \left(\frac{1.80}{2.09}\right)^{0.67} \approx 0.094.
\]
An abatement effort of approximately $9.4\%$ of output suffices to restore a finite 
steady-state pollution. With $\bar{\theta} = 0.10 > \theta^*$, the controlled 
steady-state capital falls to $\bar{K}(0.10) \approx 5.81$, and the effective 
emission flow satisfies:
\[
\gamma(1-\bar{\theta})\,A\bar{K}(\bar{\theta})^\alpha \approx 1.78 < \bar{f} = 1.80,
\]
placing the economy just inside the reversibility region. The stable steady-state 
pollution is $\bar{P}_1(0.10) \approx 5.37$, substantially higher than in the 
baseline but finite. This illustrates a general feature of the model: as the economy 
approaches the irreversibility boundary from below, the stable pollution level rises 
sharply, so that the cost of near-irreversibility is a high but finite long-run 
pollution stock.

\section{Conclusion}
This paper has embedded irreversible pollution dynamics \`{a} la Forster (1975) into a standard Ramsey optimal growth model with Cobb-Douglas production and log utility. 
We first analysed the benchmark case without abatement and established a simple threshold condition $\gamma A \bar{K}^{\alpha} > \bar{f}$ under which pollution grows without bound and the irreversible regime is optimal. 
We then introduced abatement expenditure captured by $\theta(t)$ as a control variable and studied the counterpart problem with a linear specification for the abatement control to obtain a clean and fully analytical comparison with the benchmark.
The crucial (and robust) result brought out by our analysis is that abatement is a powerful instrument for 
sustainability since it expands the set of initial conditions from which convergence to a low pollution steady state is possible. However, optimal abatement operates within the bounds set by a specific ecological-economic condition (EEC), which is essentially determined by nature, not by policy.
Even more, the (EEC) is robust to nonlinear specifications of the abatement technologies, which reinforces its preeminence with respect to the standard economic fesability conditions. This ultimately means probably that mitigation tools (like abatement) alone are not enough to avoid irreversible pollution regimes, alternative adaptation avenues may be additionally activated.

Several promising extensions suggest themselves.  First, our analysis assumes a deterministic environment. Introducing uncertainty through stochastic shocks to the absorption capacity $f(P)$ or to the emission factor $\gamma$ would allow us to study precautionary abatement and the role of risk aversion in the face of irreversibility. 
Preliminary work by Boucekkine et al. (2025) in a single-state-variable setting suggests that uncertainty can accelerate pollution accumulation as the irreversibility threshold draws near; extending this to a full Ramsey framework would be a natural next step. 
Second, we have treated abatement technology as exogenous. 
Endogenising the productivity of abatement through learning-by-doing or directed technical change (as in Amigues and Durmaz, 2019) would allow the economy to improve its abatement efficiency over time, potentially relaxing the economic feasibility constraints. 
Whether such improvements can overcome the structural ecological condition remains an open question. 
Third, our model considers a single pollutant. 
In reality, multiple pollutants interact through complex biogeochemical cycles; a multi-pollutant extension would enrich the analysis of irreversibility, particularly in the case of chemical cocktails (Persson et al., 2013).


\newpage
\appendix

\section*{APPENDIX}

\section{Proof of Theorem \ref{thm:uncontrolled_steady}}\label{appa}

For purposes of local stability analysis, we derive a linear approximation of nonlinear ODEs (7a)-(7c). Assuming
\begin{equation*}
K(t) = \bar{K} + \Delta K(t), \qquad C(t) = \bar{C} + \Delta C(t), \qquad P(t) = \bar{P} + \Delta P(t),
\end{equation*}
and using the Taylor expansion, the linearized ODE system around the steady state (8)-(10) is
\begin{align}
\Delta \dot{K} &= \rho \Delta K - \Delta C, \nonumber \\
\Delta \dot{C} &= \frac{\bar{C}}{\rho} \, \alpha(\alpha-1) A \bar{K}^{\alpha-2} \, \Delta K, \label{A1} \\
\Delta \dot{P} &= \gamma \alpha A \bar{K}^{\alpha-1} \, \Delta K - f'(\bar{P}) \, \Delta P. \nonumber
\end{align}

The linear ODEs (A1) can be written in matrix notations as $\dot{\mathbf{x}} - J\mathbf{x} = 0$, where $\mathbf{x}(t) = (\Delta K(t), \Delta C(t), \Delta P(t))^{\mathsf{T}}$ and the matrix
\begin{equation}
J = 
\begin{pmatrix}
\rho & -1 & 0 \\[4pt] 
\dfrac{\bar{C}}{\rho} \, \alpha(\alpha-1) A \bar{K}^{\alpha-2} & 0 & 0 \\[6pt]
\gamma \alpha A \bar{K}^{\alpha-1} & 0 & -f'(\bar{P})
\end{pmatrix}
\end{equation}
is the Jacobian of the nonlinear ODE system (7a)-(7c) evaluated at the equilibrium $(\bar{K}, \bar{C}, \bar{P}_i)$, $i=1,2$. Following standard analysis, the stability of the three-dimensional ODE system  around the steady states $(\bar{K}, \bar{C}, \bar{P}_i)$ is determined by the roots of the characteristic equation
\begin{equation}
\det(J - \lambda I) = 0, \label{A3}
\end{equation}
which has three real roots
\begin{align}
\lambda_1 &= -f'(\bar{P}), \label{A4} \\
\lambda_2 &= \frac{\rho - \sqrt{\rho^2 - 4a}}{2}, \\
\lambda_3 &= \frac{\rho + \sqrt{\rho^2 - 4a}}{2},
\end{align}
where
\begin{equation}
a = \frac{\bar{C}}{\rho} \, \alpha(\alpha-1) A \bar{K}^{\alpha-2}.
\end{equation}

For further analysis, let us rewrite the characteristic equation as
\begin{equation}
\bigl(f'(\bar{P}) + \lambda\bigr) \bigl(\lambda^2 - \rho\lambda + a\bigr) = 0, \label{A5}
\end{equation}
where the equation
\begin{equation}
\lambda^2 - \rho\lambda + a = 0
\end{equation}
describes the stability of the standard Ramsey model  without pollution (at $P=0$) around the steady state $(\bar{K}, \bar{C})$. This equation has two real roots $\lambda_1$ and $\lambda_2$ of opposite signs. It means that the equilibrium of the model (7a)-(7b) is a saddle point $(\bar{K}, \bar{C})$, which is a well-known economic fact. Barro et al. (Chapter 2) prove that the optimal trajectory $K(t)$, $C(t)$ in the Ramsey economy converges to the saddle steady-state pair $(\bar{K}, \bar{C})$ and even estimate the speed of convergence.

Now, let us consider the steady state $(\bar{K}, \bar{C}, \bar{P}_1)$ in the model  with pollution. The value $\bar{P}_1$ is calculated using the monotonically increasing function $f_1^{-1}(z)$ (shown in Figure 2). Then $f'(\bar{P}_1) > 0$ and the root $\lambda_1 = -f'(\bar{P}_1)$ is negative; therefore, the pollution $P(t)$ also converges to the steady-state value $\bar{P}_1$.

Similar analysis of the steady state $(\bar{K}, \bar{C}, \bar{P}_2)$ shows that then the root $\lambda_1 = -f'(\bar{P}_2)$ is positive (since $\bar{P}_2 > \hat{P}$ where $f'(\bar{P}_2)<0$). Therefore, the pollution $P(t)$ diverges from the value $\bar{P}_2$ and the steady state is unstable. \hfill$\blacksquare$

\end{document}